\documentclass[letterpaper]{IEEEtran}
\usepackage{amsfonts}
\usepackage[dvips]{graphicx}
\usepackage{epsfig,latexsym}
\usepackage{float}
\usepackage{indentfirst}
\usepackage{amsmath}
\usepackage{amssymb}
\usepackage{times}
\usepackage{subfigure}
\usepackage{stfloats,algorithm,algorithmic,bm}
\usepackage{psfrag}
\usepackage{cite}
\usepackage{subfigure}
\newtheorem{theorem}{$\mathbf{Theorem}$}
\newtheorem{lemma}[theorem]{$\mathbf{Lemma}$}

\begin{document}
\title{User Access Mode Selection in Fog Computing Based Radio Access Networks}
\author{Shi Yan, Mugen Peng,~\IEEEmembership{Senior~Member,~IEEE,} and Wenbo~Wang~\IEEEmembership{Member,~IEEE}.\\
{Key Laboratory of Universal Wireless Communication, Ministry of Education\\
Beijing University of Posts and Telecommunications, Beijing, 100876, China}}\maketitle

\begin{abstract}

Fog computing based radio access network is a promising
paradigm for the fifth generation wireless communication
system to provide high spectral and energy efficiency. With the
help of the new designed fog computing based access points (F-APs),
the user-centric objectives can be achieved through the adaptive
technique and will relieve the load of fronthaul and alleviate the
burden of base band unit pool.
In this paper, we derive the coverage probability and ergodic rate for both F-AP
users and device-to-device users by taking into account the different nodes locations,
cache sizes as well as user access modes.
Particularly, the stochastic geometry tool is used to derive expressions for above performance metrics. Simulation results validate the accuracy of our analysis and we obtain interesting tradeoffs that depend on the effect of
the cache size, user node density, and the quality of service constrains on the different performance metrics.

\end{abstract}

\section{INTRODUCTION}
The fifth generation (5G) mobile wireless system is  proposed to be initially deployed in 2020. Compared with the current  fourth generation (4G) mobile wireless system, it is given the envision of 1000 times higher wireless area capacity and is expected to save up to 90\% of
energy consumption per service compared with the current
fourth generation (4G) mobile wireless system
\cite{C1}. To achieve these
goals and alleviate the existing challenges in cloud radio access networks (C-RANs)\cite{C-RAN}\cite{H-CRAN}, the fog computing radio access network (F-RAN) has been proposed as
a new network architecture by incorporating of fog computing, edge storage
and centralized cloud computing into radio access networks \cite{F-RAN}.

Fog computing, which is similar to edge computing, is first proposed by Cisco \cite{Fog}. It extends cloud computing and services to the edge of the network. In F-RANs, services cannot only be executed in a centralized unit such as the BBU pool in C-RAN, but also can be hosted at smart terminal devices which are closer to the users. Meanwhile, through the user-centric adaptive techniques such as device-to-device (D2D), distributed coordination, and large-scale centralized
cooperation, users don't have to connect to the centralized cloud computing unit to complete the data transmission, which will relieve the load of fronthaul and alleviate the burden of BBU pool.
In order to execute the above, the traditional access point (AP) is evolved
to the fog computing based access point (F-AP) through equipped with a certain caching and sufficient
computing capabilities to execute the local cooperative signal processing in the physical layer.

Many previous works have been done to analyze the ergodic rate performance of C-RAN systems. In \cite{b9}, the ergodic rate of distributed remote radio heads (RRHs)
is characterized in C-RAN with spatially single antenna
random locations and the minimum number of RRHs for the desired user to meet a predefined quality of service is analyzed.
In \cite{b10}, it is demonstrated that the large scale
fading exponent has a significant impact on the capacity of large
C-RAN systems.

However, the fronthaul limitation between RRHs and Cloud is a remarkable challenge to block the commercial practices.
Consequently, taking advantage of the fog-computing to switch the content cache or data process to users is the key to improve both spectrum and energy efficiency as well as can relieve the load of fronthaul in cloud computing based network architectures.
D2D communications as an adaptive technology not only can provide the efficient utilization of available radio resources to improve the connectivity of devices, but also has the ability of supporting proximity-based services, such as social networking applications and content sharing. Existing research on D2D networks is mainly focused on underlaid cellular networks with fixed location model \cite{b12} $\sim$ \cite{b14}. Different from the traditional underlaid cellular networks, the main challenge
to analyze the rate performance in F-RAN is that the F-APs are often deployed randomly and the severe intra-tier and inter-tier interference may drastically deteriorate the performance of both the D2D and F-AP users.

In this paper, we analyze the coverage probability and ergodic rate with
three user access modes in a downlink F-RAN,
where both the F-AP nodes and D2D users are modeled as spatial Poisson Point Process (PPP)
distribution. The
contributions are three-folds.
\begin{itemize}
\item The ergodic rates of D2D mode, nearest F-AP mode, local distributed coordination mode, and coverage probability of the first two modes in F-RAN system are
characterized, where both the intra-tier and inter-tier
interference and distributed cache are considered.

\item The closed-form expressions for the ergodic rate are presented in some special cases,
which can make the analysis not only tractable, but also flexible. Moreover, based on the proposed performance metrics, the impacts of the cache size, user node density, and the quality of service (QoS) constrains are
characterized.

\item The Monte Carlo simulation results evaluate impacts of the F-AP nodes density, SIR threshold, cache size and association schemes on the ergodic rate. And an adaptation user access mode
selection mechanism is proposed to improve F-RAN system performances.
\end{itemize}


\section{SYSTEM MODEL}
\subsection{F-RAN System Model}

A F-RAN downlink system is considered in this paper,
where a group of F-APs are deployed according to a two-dimensional PPP $\Phi_f$
with density of $\lambda_f$ in a disc plane ${\cal{D}}^2$. As the new designed AP,
F-AP integrates not only the front radio frequency (RF) but also the physical processing
functionalities and procedures of the upper layers, which made F-AP has a sufficient
computing capabilities to execute the local cooperative signal processing in the physical
layer and implement the caching resource management.

We assume the spatial distribution of users is
modeled as an independent PPP $\Phi_u$ distribution with constant intensity
$\lambda_u$. By setting $p\sim(0,1]$ as the probability that a user support
direct connection to other intelligent terminal users, i.e., D2D,
the distribution of the D2D users location can be denoted as a
thinning homogeneous $\Phi_{du}$ with the density of $\lambda_{du}=p\lambda_u$.
Meanwhile according to Marking Theorem, the distribution of F-AP users follows a
stationary PPP $\Phi_{fu}$ with the density of $\lambda_{fu}=(1-p)\lambda_u$.
Without any loss of generality,
each F-AP and D2D user is assumed as single antenna configuration with
a fixed transmission power, defined as $P_f$ and $P_d$, respectively,
and our analysis is focus on a desired user (denoted by $U$) located at
the origin of the disc ${\cal{D}}^2$.


\subsection{Cache Model}
In this paper, we consider there are $N$ video contents in the network, and all the
video contents are assumed to have the same distribution. Each of
the D2D user and F-AP has a limited caching storage space with the size of $C_d$ and $C_f$, respectively,
and $C_d<C_f<N$.

In previous research, it has been found that people are always interested in
the most popular video contents \cite{Zipf}. In other word, only a small
portion of the $N$ contents are frequently accessed by the majority
of users. Therefore, the demand probability of the $i$-th popularity video content
can be modeled as the following Zipf distribution

\begin{equation}\label{popu}
{f_i(\sigma, N)} = \frac{{1/{i^\sigma }}}{{\sum\nolimits_{k = 1}^N {1/{k^\sigma }} }},
\end{equation}
where the video content with a smaller index has a larger probability of being requested by users, i.e. $f_i(\sigma, N)>f_j(\sigma, N)$, if $i<j$.
 Zipf exponent $\sigma>0$ controls
the relative popularity of files, and with the larger $\sigma$ the caching storage has a fewer of popular video contents accounting for the majority of the requests.

Content caching probability is defined as the probability of an event that the desired user $U$ can find it requested video content $V$ in its corresponding caching, i.e., $p_c^x={\Pr }(V \in C_x)$, where $x$ denotes the node of user accessed. By setting the caching storages in D2D users and F-APs can only cache the most popularly requested video contents, the content caching probabilities of each D2D user and F-AP can be respectively denoted as

\begin{equation}\label{ccpd}
p^{D}_c={\Pr }(V \in C_d)={\sum\nolimits_{i=1}^{C_d}}{f_i(\sigma_d, N)}
\end{equation}
\begin{equation}\label{ccpap}
p^{F}_c={\Pr }(V \in C_f)={\sum\nolimits_{i=1}^{C_f}}{f_i(\sigma_f, N)}.
\end{equation}

As \eqref{ccpd} and \eqref{ccpap} shown, the same kind of node has the same content caching probability. In other words, each D2D user in area ${\cal{D}}^2$ store the same content cache, and video contene cached in different F-APs are also the same.

\subsection{User Access Modes}
In this paper, we consider users will access to the F-RAN by three user-centric access modes according to users' communication distance, content caching probability and the QoS requirements, named: D2D mode, nearest F-AP mode and local distributed coordinated mode. Let $U \to X$ signify that desired user $U$ is associated to a node located at $X$, and ${\left\| X\right\|}$ denotes the distance between $U$ and $X$.


\begin{itemize}
\item \textbf{D2D mode:} D2D mode is enabled when the desired user $U$ support D2D mode and it can successfully obtain the requested contents from another D2D user in a known location within a distance threshold $L_d$ meanwhile the signal-to-interference ratio (SIR) $\gamma_d$ between the two D2D users is larger than a pre-set SIR threshold $T_d$. Thus,
\end{itemize}
\begin{equation}\label{d2d}
\begin{gathered}
    \Psi_{D} =\{ X_d: X_d \in \Phi_{du} , {\left\| X\right\|} \leq L_d, \hfill\\
   \quad\quad\quad V \in C_d, \gamma_d \geq T_d  \}.\hfill
\end{gathered}
\end{equation}

\begin{itemize}
\item \textbf{ Nearest F-AP mode:} When the desired user $U$ does not support D2D mode, or $U$ support D2D mode but the requested content $V$ is not cached in its nearby D2D user or the SIR $\gamma_d$ is not achieved the SIR threshold $T_d$. Thus, $U$ try to access its nearest F-AP node which can respond to the desired user's content request, and the SIR $\gamma_f$ between them is larger than SIR threshold $T_f$. The associated F-AP for user $U$ can be obtained as:
\end{itemize}
\begin{equation}\label{1fap}
\begin{gathered}
    \Psi_{F}= \{ X_f:  \mathop {\arg \min }\limits_{X \in {\Phi _{f}}} \left( {\left\| {{X}} \right\|} \right), V \in C_f, \gamma_f \geq T_f, \hfill\\
    \quad\quad\quad U \notin \Phi_{du} \cup V \notin C_d \cup \gamma_d < T_d,  \}.\hfill
\end{gathered}
\end{equation}

\begin{itemize}
\item \textbf{ Local distributed coordination mode:} The local distributed coordination mode means that the desired user $U$ associates to multiple F-APs near to it in a user-centric cluster with a radius $L_c$. In this mode F-RAN can adjust the value of $L_c$ to satisfy well with the requirement of video content and the quality of SIR.
\end{itemize}
\begin{equation}\label{cluster}
    \Psi_{C}= \{ X_c:  X \in \Phi_{f}, \forall X \in B(U,L_c) \cap \Phi_{f} \},\hfill
\end{equation}
where $B(a, b)$ denote all the point in a circle centered in $a$ with radius $b$.

\subsection{Signal-to-Interference Ratio}

In this paper, we focus on the interference-limited scenario
since the interference is much larger than the noise, i.e., the noise can be neglected
and the SIR is dominated.
Path loss is represented by ${\left\| X \right\|^{ - \alpha }}$, where $\alpha>2$ is the path loss exponent. We denote $\alpha_d$ as the path loss exponent for D2D user link and $\alpha_f$ as the path loss exponent for F-AP to common user or F-AP to D2D user link.

Then, if $U$ is served by a D2D user which has a fixed distance of $\left\| {{X_d}} \right\|$ to $U$, the received SIR at the desired user is given by

\begin{equation}\label{SIRd}
SIR(U \to {X_d}){\rm{ }} = \gamma_d= \frac{{{P_d}{h_d}{{\left\| {{X_d}} \right\|}^{ - \alpha_d}}}}{{{I_{d,du}} + {I_{f,du}}}},
\end{equation}
where $h_d$ characterize the flat Rayleigh channel fading between two D2D users, and ${{\left\| {{X_d}} \right\|}^{ - \alpha_d }}$ denotes the path loss.
${I_{d,du}} = {\sum _{i \in {\Phi _{du}}/{d}}}{P_d}{g_i}r_i^{ - \alpha_d }$ denotes interference from other D2D users, $ g_i \sim \exp (1)$ and $r_i^{ - \alpha_d }$ denote the exponentially distributed fading power over the Rayleigh fading channel and path loss from other D2D user to $U$, respectively. ${I_{f,du}} = {\sum _{j \in {\Phi _f}}}{P_f}{g_j}l_j^{ - \alpha_f }$ denotes inter-tier interference from F-APs, the definition of ${g_j}$ and $l_j^{ - \alpha_f }$ are similar to that in ${I_{d,du}}$.

Next, if $U$ is served by a single nearest F-AP in nearest F-AP mode, the SIR is given by

\begin{equation}\label{SIRf}
SIR(U \to {X_f}){\rm{ }} = \gamma_f= \frac{{{P_f}{h_f}{{\left\| {{X_f}} \right\|}^{ - \alpha_f }}}}{{{I_{f,fu}} + {I_{d,fu}}}},
\end{equation}
where ${I_{f,fu}} = {\sum _{{i'} \in {\Phi _f}/f}}{P_f}{g_i'}l_{i'}^{ - \alpha_f }$, and ${I_{d,fu}} = {\sum _{{j'} \in {\Phi _{du}}}}{P_d}{g_j'}r_{j'}^{ - \alpha_f }$ denote the intra-tier interference from other F-APs and inter-tier interference from D2D users, respectively.

Finally, if the local distributed coordination mode is selected, $U$ is not only
served by the single nearest F-AP, but several potential F-APs which form a F-AP cluster in \eqref{cluster}. The F-AP cluster formation
leads to that the received signal of the desired user is a sum form and after considering the inter-tier and intra-tier interference, the received SIR of $U$ in local distributed coordination mode can be given by

\begin{equation}\label{SIRc}
SIR(U \to {X_c}){\rm{ }} = \gamma_c= \frac{{\sum _{{c} \in {\Psi _{C}}}{P_f}{h_c}{{\left\| {{X_c}} \right\|}^{ - \alpha_f }}}}{{{I_{f,cu}} + {I_{d,fu}}}},
\end{equation}
where the intra-tier interference from the F-APs out of the cluster denotes as ${I_{f,cu}} = {\sum _{{v} \in {\Phi _f}/{\Psi _C}}}{P_f}{g_v}l_{v}^{ - \alpha_f }$, and the inter-tier interference expression $I_{d,fu}$ is equal to that in \eqref{SIRf}.

\section{PERFORMANCE ANALYSIS}

In this section, we derive the coverage probability and ergodic rate for F-RAN with three different association modes. The ergodic rate is defined as $R_x=p_x\mathbb{E}\left[\rm{ln}\left(1+SIR(U \to X_x)\right)|SIR(U \to X_x)>T_x\right]$, where $p_x$ denote the probability of the desired user $U$ select the $x$ mode and the unit of the ergodic rate is in terms of nats/s/Hz. $\mathbb{E}(\cdot)$ is the expectation respect to the channel fading distribution as well as the locations of the random transmitter nodes.

\subsection{D2D mode}
In D2D communications, a direct link is established between the desired user $U$ and its service D2D user which has a known location $X_d$. The probability of $X_d$ located in distance threshold $L_d$ meanwhile has the requested content $V$ can be given as
\begin{equation}\label{D2Dp}
p_D=p(1-\exp({-\pi \lambda_{du} p_c^D L_d^2})).
\end{equation}

\begin{proof}
By using the property of 2-D Poisson process, the probability distribution of the nodes number $m$ in a circle area $\pi l^2$ with radius limit $l$ can be derived as

\begin{equation}\label{PM}
{\rm{P_r}}\left\{ {\Phi \left( {\pi {l^2}} \right) = {m}} \right\}
 = \frac{{{{\left( {{\lambda_X}\pi {l^2}} \right)}^{{m}}}{e^{ - {\lambda_X}\pi {l^2}}}}}{{\left( {{m}} \right)!}}.
\end{equation}

Let $l=L_d$, $\lambda_X=p_c^D \lambda_{du}$ and $m=0$. Then, we have the probability of none D2D user has the the requested video content $V$ within the distance limit $L_d$. Therefore, \eqref{D2Dp} can be given as the probability of complementary events and multiply the probability of $U$ support D2D mode.
\end{proof}

The probability of SIR $\gamma_d$ between the two D2D users larger than a SIR threshold $T_d$ is also called coverage probability. And in D2D mode

 \begin{equation}\label{PD}
\begin{array}{l}
 P_D(T_d, \alpha_f, \alpha_d, \left\| {{X_d}} \right\|) = \Pr \left( {\frac{{{P_d}{h_d}{{\left\| {{X_d}} \right\|}^{ - {\alpha _d}}}}}{{{I_{d,du}} + {I_{f,du}}}} \ge {T_d}} \right) \\
  = \Pr \left( {{h_d} \ge \frac{{{T_d}{{\left\| {{X_d}} \right\|}^{  {\alpha _d}}}}}{{{P_d}}}\left( {{I_{d,du}} + {I_{f,du}}} \right)} \right) \\
 \mathop  = \limits^{\left( a \right)} \mathbb{E}\left[ {\exp \left( { - \frac{{{T_d}{{\left\| {{X_d}} \right\|}^{  {\alpha _d}}}}}{{{P_d}}}\left( {{I_{d,du}} + {I_{f,du}}} \right)} \right)} \right] \\
 \mathop  = \limits^{\left( b \right)} {L_{{I_{d,du}}}}\left( {\frac{{{T_d}{{\left\| {{X_d}} \right\|}^{  {\alpha _d}}}}}{{{P_d}}}} \right){L_{{I_{f,du}}}}\left( {\frac{{{T_d}{{\left\| {{X_d}} \right\|}^{  {\alpha _d}}}}}{{{P_d}}}} \right) \\
=\exp \left( { - \pi {{\left\| {{X_d}} \right\|}^{\frac{{2{\alpha _d}}}{{{\alpha _f}}}}}\left( {{\lambda _{du}} + {{\left( {\frac{P_f}{P_d}} \right)}^{\frac{2}{{{\alpha _f}}}}}{\lambda _f}} \right)C\left( {{\alpha _f}} \right)T_d^{\frac{2}{{{\alpha _f}}}}} \right), \\
 \end{array}
 \end{equation}
where (a) follows from the Laplace transform of $h_d \sim \exp(1)$ and the independence of $I_{d,du}$ and $I_{f,du}$ \cite{La}\cite{La2}. (b) follows from letting $s = {{{{T_d}{{\left\| {{X_d}} \right\|}^{ - {\alpha _D}}}}}/{{{P_d}}}}$ in the Laplace transforms of $I_{d,du}$ and $I_{f,du}$, $C\left( \alpha_f  \right) = {{2{\pi}\csc \left( {{{2\pi }}/{\alpha_f }} \right)}}/{\alpha_f }$.

Then, the ergodic rate for D2D mode under the conditions in \eqref{d2d} can be derived as
 \begin{equation}\label{RD}
\begin{gathered}
 {R_d} = p_D\mathbb{E} \left[ {\ln \left( {1 + {\gamma _d}} \right)} | \gamma_d \ge T_d\right] \hfill\\
 \mathop  \approx \limits^{\left( a \right)}  p_D\ln(T_d)P_D(T_d, \alpha_f,\alpha_d \left\| {{X_d}} \right\|) - \frac{{p_D{\alpha _f}}}{2}\hfill\\
\cdot{\rm{Ei}}\left[ { -T_d^{\frac{2}{\alpha_f}} \pi {{\left\| {{X_d}} \right\|}^{\frac{{2{\alpha _d}}}{{{\alpha _f}}}}}\left( {{\lambda _{du}} + {{\left( {\frac{P_f}{P_d}} \right)}^{\frac{2}{{{\alpha _f}}}}}{\lambda _f}} \right)C\left( {{\alpha _f}} \right)} \right], \hfill\\
\end{gathered}
 \end{equation}
where $(a)$ follows in the high SIR conditions ${\ln}\left( {1 + \gamma_d} \right) \to {\ln }\left( \gamma_d \right)$, Ei$[s]=-\int_{-s}^\infty{{e^{-t}}/{t}\rm{d}t}$ is the exponential integral function.
\begin{proof}
See Appendix A.
\end{proof}

\subsection{Nearest F-AP mode}

Next, we focus on the nearest F-AP mode which will be triggered if the desired user $U$ cannot meet the conditions of D2D mode, i.e.,
\begin{equation}\label{1FAPP}
p_F=1-p_DP_D(T_d, \alpha_f, \alpha_d, \left\| {{X_d}} \right\|).
 \end{equation}

In this mode, the desired user will try to access its nearest F-AP $X_f$ which has the requested content $V$. The probability density function (PDF) of the distance between $X_f$ and $U$ can be derived by using a similar way as \eqref{D2Dp}
 \begin{equation}\label{cdfX}
\begin{gathered}
 {f_{\left\| {{X_f}} \right\|}}\left( {{r_f}} \right) = \frac{{\partial \left( {1 - \Pr \left( {{\rm{No\;F - AP\;closer\;than\; }}{r_f}} \right)} \right)}}{{\partial {r_f}}} \hfill\\
  = \frac{{\partial \left( {1 - \exp \left( { - \pi {\lambda _f}p_c^Fr_f^2} \right)} \right)}}{{\partial {r_f}}} = 2\pi {\lambda _f}p_c^F{r_f}{e^{ - \pi {\lambda _f}p_c^Fr_f^2}}. \hfill\\
 \end{gathered}\
 \end{equation}

 Thus, the coverage probability of nearest F-AP mode can be calculated as

 \begin{equation}\label{PF}
 \begin{array}{l}
 {P_F}\left( {{T_f},\alpha_f, p_c^F} \right) = \Pr \left( {\frac{{{P_f}{h_f}{{\left\| {{X_f}} \right\|}^{ - {\alpha _f}}}}}{{{I_{f,fu}} + {I_{d,fu}}}} \ge {T_f}} \right)\\
  \quad= \int_0^\infty  {{\rm{Pr}}\left( {{h_f} \ge \frac{{{T_f}r_f^{{\alpha _f}}}}{{{P_f}}}\left( {{I_{f,fu}} + {I_{d,fu}}} \right)} \right)} {f_{\left\| {{X_f}} \right\|}}\left( {{r_f}} \right){\rm{d}}{r_f} \\
 \quad\mathop  = \limits^{\left( a \right)} \int_0^\infty  {{L_{{I_{f,fu}}}}\left( {\frac{{{T_f}r_f^{{\alpha _f}}}}{{{P_f}}}} \right){L_{{I_{d,fu}}}}\left( {\frac{{{T_f}r_f^{{\alpha _f}}}}{{{P_f}}}} \right)} {f_{\left\| {{X_f}} \right\|}}\left( {{r_f}} \right){\rm{d}}{r_f} \\
 \quad \mathop  = \limits^{\left( b \right)} \int_0^\infty  {\exp \left( { - \pi {\lambda _f}r_f^2p_c^F\rho \left( {{T_f},{\alpha _f}} \right)} \right)} \\
  \cdot \exp \left( { - \pi {\lambda _{du}}r_f^2C\left( {{\alpha _f}} \right){{\left( {\frac{{{P_d}{T_f}}}{{{P_f}}}} \right)}^{\frac{2}{\alpha _f}}}} \right)2\pi p_c^F{\lambda _f}{r_f}{e^{ - \pi p_c^F{\lambda _f}r_f^2}}{\rm{d}}{r_f}\quad \quad  \\
 \quad = \frac{1}{{1 + \rho \left( {{T_f},{\alpha _f}} \right) + \frac{{{\lambda _{du}}}}{{{p_c^F\lambda _f}}}C\left( {{\alpha _f}} \right){{\left( {\frac{{{P_d}{T_f}}}{{{P_f}}}} \right)}^{2/{\alpha _f}}}}},\\
 \end{array}
 \end{equation}
where (a) follows the setting of $h_f \sim \rm{exp}(1)$ and the independence between inter-floor interference $I_{d,fu}$ and intra-floor interference $I_{f,fu}$; equation (b) follows the definition of the Laplace transform, and  $\rho \left( {T_f,\alpha_f } \right) = \int_{{T^{ - \frac{2}{\alpha_f }}}}^\infty  {\frac{{{T^{2/\alpha_f }}}}{{1 + {v^{\alpha_f /2}}}}{\rm{d}}v}$.

And the ergodic rate for nearest F-AP mode under the conditions in \eqref{1fap} can be given as
 \begin{equation}\label{R1F}
  \begin{gathered}
   {R_f} = p_F\mathbb{E} \left[ {\ln \left( {1 + {\gamma _f}} \right)} | \gamma_f \ge T_f\right] \hfill\\
\approx  \int_{\ln(T_f)}^{\infty}  p_F{P_F(e^\theta, \alpha_f, p_c^F)} {\rm{d}}\theta+p_F\ln(T_f)P_F(T_f, \alpha_f, p_c^F). \hfill\\
  \end{gathered}
 \end{equation}

\textbf{Special Case}: {Path loss exponent for F-AP to user link is 4 ($\alpha_f=4$), and SIR threshold $T_f > 1$}

A closed-form approximate expression can be derived in this special case, and we give the ergodic rate as Lemma 1.

\begin{lemma}
The ergodic rate for nearest F-AP mode with $\alpha_f$= 4 and $T_f > 1$ can be expressed as
\begin{equation}\label{R1FS}
  \begin{gathered}
   {R_f}^{\alpha_f=4} = p_F\mathbb{E} \left[ {\ln \left( {1 + {\gamma _f}} \right)} | \gamma_f \ge T_f\right] \hfill\\
\approx  \frac{4p_F}{\pi\sqrt{T_f}\left(1+\frac{\lambda_{du}}{p_c^F\lambda_f}\sqrt{\frac{P_d}{P_f}}\right)}
+\frac{2p_F\ln(T_f)}{\pi\sqrt{T_f}\left(1+\frac{\lambda_{du}}{p_c^F\lambda_f}\sqrt{\frac{P_d}{P_f}}\right)} \hfill\\
=\frac{2p_F(2+\ln(T_f))}{\pi\sqrt{T_f}\left(1+\frac{\lambda_{du}}{p_c^F\lambda_f}\sqrt{\frac{P_d}{P_f}}\right)}.\hfill\\
  \end{gathered}
 \end{equation}
\end{lemma}

\begin{proof}
See Appendix B.
\end{proof}

\subsection{Local distributed coordination mode}

Lastly, we discuss the local distributed coordination mode that the desired $U$
associates to multiple F-APs near to it in a user-centric cluster with a distance threshold $L_c$. This mode is triggered if $U$ don't meet the conditions of both the first two modes, i.e.,

\begin{equation}\label{ClusterP}
\begin{gathered}
p_C=(1-p_DP_D(T_d, \alpha_f, \alpha_d, \left\| {{X_d}} \right\|))\hfill\\
\quad\;\;\cdot (1-P_F(T_f, \alpha_f, p_c^F)\hfill\\
\quad\;\;=p_f(1-P_F(T_f, \alpha_f, p_c^F)).\hfill\\
\end{gathered}
 \end{equation}

To extend to access the arbitrary F-APs, the main problem here is the
exact PDF expression of ${{\sum _{{v} \in {\Psi _{C}}}{P_f}{h_v}{{\left\| {{X_v}} \right\|}^{ - \alpha_f }}}}$ is very difficult to obtain, so we tried to get around this
problem by following the Lemma 1 in \cite{lncluster}.

For any $A>0$
\begin{equation}\label{log}
\mathbb{E}[{\ln}\left( {1 + A} \right)] = \int_0^\infty  {\frac{1}{z}} \left( {1 - {e^{ - Az}}} \right){e^{ - z}}dz
\end{equation}

Therefore, with \eqref{log}, the ergodic rate expression can be written as \eqref{RC}
on the top of next page,
\begin{figure*}[ht]
 \begin{equation}\label{RC}
\begin{array}{l}
 {R_c}{\rm{ = }}{p_C}{\mathbb{E}}\left[ {\ln \left( {1 + \frac{{\sum\nolimits_{c \in {\Psi _C}} {{P_f}{h_c}{{\left\| {{X_c}} \right\|}^{ - {\alpha _f}}}} }}{{{I_{f,cu}} + {I_{d,fu}}}}} \right)} \right]
  ={p_C} \mathbb{E}\left[ {\int_0^\infty  {\frac{{{e^{ - z}}}}{z}\left( {1 - \exp \left( { - \frac{{z\sum\nolimits_{c \in {\Psi _C}} {{P_f}{h_c}{{\left\| {{X_c}} \right\|}^{ - {\alpha _f}}}} }}{{{I_{f,cu}} + {I_{d,fu}}}}} \right)} \right){\rm{d}}z} } \right] \\
 \mathop {\rm{ = }}\limits^{\left( a \right)} {p_C}{{\mathbb{E}}_{\Phi ,h,g}}\left[ {\int_0^\infty  {\frac{1}{s}\exp \left( { - s\left( {{I_{f,cu}} + {I_{d,fu}}} \right)} \right)\left[ {1 - \exp \left( { - s\sum\nolimits_{c \in {\Psi _C}} {{P_f}{h_c}{{\left\| {{X_c}} \right\|}^{ - {\alpha _f}}}} } \right)} \right]} {\rm{d}}s} \right] \\
 \mathop  = \limits^{\left( b \right)} {p_C}\int_0^\infty  {\frac{1}{s}\left\{ {{{\mathbb{E}}_{\Phi ,h,g}}\left[ {\exp \left( { - s\left( {{I_{f,cu}} + {I_{d,fu}}} \right)} \right) - \exp \left( { - s\left( {{I_{d,fu}} + \sum\nolimits_{c \in {\Phi _f}} {{P_f}{h_c}{{\left\| {{X_c}} \right\|}^{ - {\alpha _f}}}} } \right)} \right)} \right]} \right\}} {\rm{d}}s \\
  = {p_C}\int_0^\infty  {\frac{1}{s}\left\{ {{L_{{I_{f,cu}}}}\left( s \right){L_{{I_{d,fu}}}}\left( s \right) - {L_{{I_{d,fu}}}}\left( s \right){L_{\sum\nolimits_{c \in {\Phi _f}} {{P_f}{h_c}{{\left\| {{X_c}} \right\|}^{ - {\alpha _f}}}} }}\left( s \right)} \right\}} {\rm{ds}} \\
  = {p_C}\int_0^\infty  {\frac{1}{s}} \exp \left( { - \pi {\lambda _{du}}C\left( \alpha_f  \right){{\left( {{P_d}s} \right)}^{\frac{2}{{{\alpha _f}}}}}} \right)\left[ {\exp \left( { - 2\pi p_c^F{\lambda _f}\int_R^\infty  {\frac{{{P_f}sv}}{{{v^{{\alpha _f}}} + {P_f}s}}} {\rm{d}}v} \right) - \exp \left( { - \pi p_c^F{\lambda _f}C\left( \alpha_f  \right){{\left( {{P_f}s} \right)}^{\frac{2}{{{\alpha _f}}}}}} \right)} \right]{\rm{d}}s, \\
 \end{array}
   \end{equation}
   \hrulefill
\end{figure*}
where (a) follows set $s=z\cdot(I_{f,cu}+I_{d,fu})$, (b) follows ${{c} \in {\Psi _{C}}} \cup \{ {\Phi _f}/{\Psi _C}\} ={\Phi _f}$.

\subsection{Adaptation user access mode selection mechanism}

In addition, based on communication distance, nodes location, SIR QoS requirements and caching capabilities, an adaptation mode selection mechanism is presented in this subsection to take full advantages of these three modes.

\begin{algorithm}[]\tiny
   \caption{Adaptation User Access Mode Selection mechanism}
   \begin{algorithmic}[1]
   \small\STATE \textbf{Initialize}  $\Psi_{D}= \emptyset ,\Psi_{F}= \emptyset$, $\Psi_{C}= \emptyset$.\\
    \STATE \textbf{Step 1} Check the cache content of another D2D user nearby the desired user $U$ with a radius threshold $L_d$, $B(U,L_d)\cap \Phi_{du}=\{X_1,X_2,...,X_D\}$. \\
   \STATE \quad\textbf{for} $i= 1,2,...,D$ do. \\
   \STATE \quad\quad\textbf{if} $V \in C_d^i$
   \STATE \quad\quad\quad Calculate the SIR $\gamma_d^i$ from \eqref{SIRd}.\\
   \STATE \quad\quad\quad\quad\textbf{if} $\gamma_d^i \ge T_d$ .\\
   \STATE \quad\quad\quad\quad$X_i \in \Psi_{D}$ User select \textbf{D2D mode} with $X_i$. \textbf{break}
   \STATE \quad\quad\quad\quad\textbf{end if}
   \STATE \quad~~\textbf{end if}
   \STATE \quad\textbf{end for}
   \STATE \textbf{Step 2} Find the nearest F-AP of the desired user $U$ in $\Phi_{f}=\{X_1,X_2,...,X_F\}$.\\
   \STATE  \quad\textbf{Set} $\left\| {{X_f}} \right\|=\left\| {{X_1}} \right\|$.
   \STATE \quad\quad\textbf{for} $j= 2,3,...,F$ do. \\
   \STATE \quad\quad\quad\textbf{if} $\left\| {{X_j}} \right\|<\left\| {{X_f}} \right\|$ .\\
   \STATE \quad\quad\quad Nearest node $\left\| {{X_f}} \right\|=\left\| {{X_j}} \right\|$ .\\
   \STATE \quad\quad\quad\textbf{end if}
   \STATE \quad\quad\textbf{end for}
   \STATE \quad\textbf{if} $V \in C_f^f$
   \STATE \quad Calculate the SIR $\gamma_f$ with $\left\| {{X_f}} \right\|$ from \eqref{SIRf}.\\
   \STATE \quad\quad\textbf{if} $\gamma_f \ge T_f$ .\\
   \STATE \quad\quad$X_f \in \Psi_{F}$ User select \textbf{1 FAP mode} with $X_f$.
   \STATE \quad\quad\textbf{else}
   \STATE \quad\quad\quad go to \textbf{Step 3}\\
   \STATE \quad\quad\textbf{end if}
   \STATE \quad\textbf{end if}
   \STATE \textbf{Step 3} User select \textbf{local distributed coordination mode} with $\Psi_{F}=B(U,L_f)\cap \Phi_{f}=\{X_1,X_2,...,X_F\}$. \\
   \end{algorithmic}
\end{algorithm}

\section{NUMERICAL RESULTS}

In this section, the accuracy of the above ergodic rate expressions
and the impact of $\lambda$, $C$ and $T$ on rate performance are
evaluated by using Matlab with Monte Carlo simulation method.
The simulation parameters are listed as follows in Table I.

\begin{table}[!htp]
\renewcommand{\arraystretch}{1.3}
\caption{SIMULATION PARAMETERS} \label{table_example}
\centering
\begin{tabular}{{c|c}}
\hline
\bfseries Parameters & \bfseries Value \\
\hline
Number of video content $N$& $1000$\\
\hline
Caching size of D2D user $C_d$& $50$\\
\hline
Caching size of F-AP $C_d$& $200 \sim 800$\\
\hline
Intensity of D2D users $p\lambda_u$ & $1 \times 10^{-3}$\\
\hline
Intensity of F-AP nodes $\lambda_f$ & $1\times 10^{-4} \sim 1\times 10^{-3}$\\
\hline
Path loss exponent $\alpha$  & 4  \cite{tc}  \\
\hline
D2D user Zipf exponent $\sigma_d$ & 0.8 \\
\hline
F-AP Zipf exponent $\sigma_f$ & 1  \\
\hline
Transmit power of D2D user $P_d$ & 3dBm \cite{Li}\\
\hline
Transmit power of F-AP $P_f$ & 23dBm \\
\hline
D2D distance threshold $L_c$ & 16m\\
\hline
Cluster distance threshold $L_c$ & 45 $\sim$ 70m\\
\hline
\end{tabular}\vspace*{-1em}
\end{table}

\begin{figure}[!htp]
\centering
\includegraphics[width=3in]{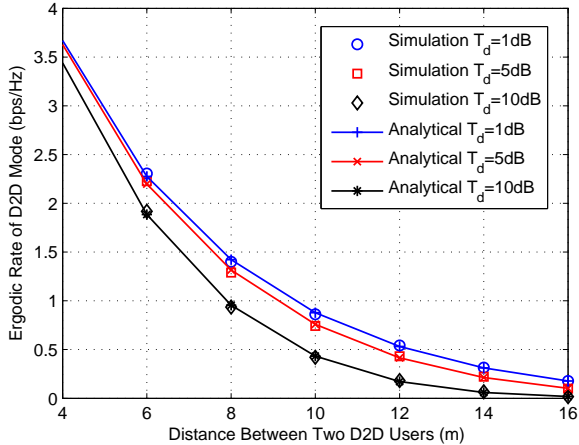}
\caption{Ergodic rate of D2D mode with different SIR thresholds versus distance between two D2D users.}
\label{f1}\vspace*{-1em}
\end{figure}

Fig. 1 shows the ergodic rate achieved by the D2D mode with the varying distance between D2D pairs in the cases of different SIR QoS thresholds $T_d$.
The analytical results closely match with the
corresponding simulation results, which validates our analysis
in Section III. It can be observed that the ergodic rate of D2D mode decreases as the distance between D2D pair increases. Similarly, the larger $T_d$ suggests that the D2D user is more strict in the quality of SIR, which leads to less user select D2D mode, and thus D2D mode ergodic rate decreases.

\begin{figure}[!htp]
\centering
\includegraphics[width=3in]{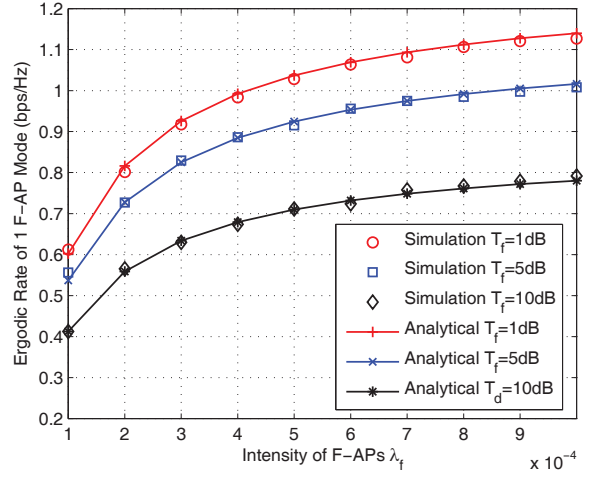}
\caption{Ergodic rate of nearest F-AP mode with different SIR thresholds versus intensity of F-AP nodes $\lambda_f$.}
\label{f1}\vspace*{-1em}
\end{figure}

Ergodic rate of nearest F-AP mode with different SIR thresholds versus intensity of F-AP nodes $\lambda_f$ is shown in Fig. 2.
As we see from Fig. 2,
the ergodic rate of nearest F-AP mode grows as the nodes intensity $\lambda_f$ increases until it reaches its upper bound. It is worth noting that the ergodic rate of nearest F-AP mode seems lower than the D2D mode, this is because the users which have good SIR are more inclined to selected D2D mode. On the other hand, when the SIR threshold increases, the ergodic rate of nearest F-AP mode follows similar trend as D2D mode.

\begin{figure}[!htp]
\centering
\includegraphics[width=3in]{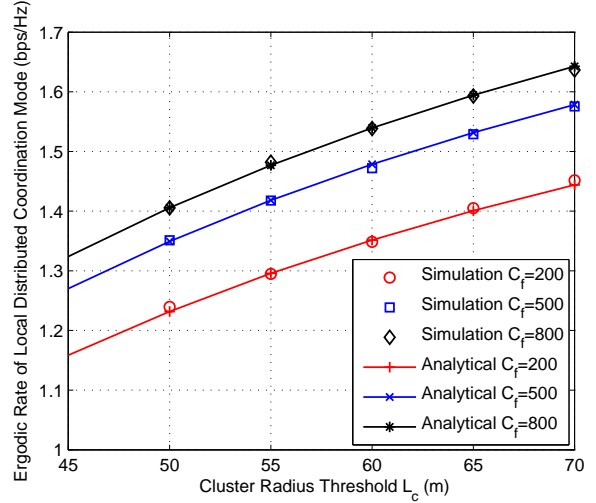}
\caption{Ergodic rate of local distributed coordination mode with different F-AP cache size versus cluster radius threshold $L_c$, $\lambda_f=2\times 10^{-4}$.}
\label{f1}\vspace*{-1em}
\end{figure}

In Fig. 3, we compare the ergodic rate of local distributed coordination mode with different F-AP cache sizes versus cluster radius threshold $L_c$. Since the number of F-APs in the cluster increases with $L_d$, it means more F-APs are serving the desired user with the increase of signal and the decline of interference. And it leads to the enlargement of cluster SIR and results in the improvement of
ergodic rate of local distributed coordination mode. Similarly, the larger cache size of F-AP $C_f$ suggests that there are more opportunities for the desired user to get the video content it needs, which leads to a higher ergodic rate.

\section{CONCLUSION}
In this paper, the expressions of coverage probability and ergodic rates in downlink
F-RAN system under three user access modes have been derived.
The impact of the cache size, SIR threshold cluster radius and nodes intensity on the coverage probability and ergodic rate are researched. Moreover, an adaptation user access mode selection mechanism is proposed to improve F-RAN system performances.

\section{ACKNOWLEDGEMENT}
This work is supported in part by the National Natural Science Foundation
of China (Grant No. 61222103 and 61361166005), the National Basic Research Program
of China (973 Program) (Grant No. 2013CB336600).

\appendix
\subsection{Proof of $R_d$}
For a positive continuous random variable $A$, we can use the following formula for computing its expectation

\begin{equation}\label{E1}
\begin{array}{l}
\mathbb{E}\left[ {A\left| {A \ge W} \right.} \right] \\
  = \int_W^\infty  {t{f_A}\left( t \right)} {\rm{d}}t = \int_W^\infty  {\int_0^t {{f_A}\left( t \right)} {\rm{d}}a} {\rm{d}}t \\
  = \int_0^W {\int_W^\infty  {{f_A}\left( t \right)} } {\rm{d}}t{\rm{d}}a + \int_W^\infty  {\int_a^\infty  {{f_A}\left( t \right)} } {\rm{d}}t{\rm{d}}a \\
  = \int_0^W {\Pr \left( {A \ge W} \right)} {\rm{d}}a + \int_W^\infty  {\Pr \left( {A \ge a} \right)} {\rm{d}}a \\
  = W\Pr \left( {A \ge W} \right) + \underbrace{\int_W^\infty  {\Pr \left( {A \ge a} \right)} {\rm{d}}a}_{S}.\\
\end{array}
 \end{equation}

Next, we focus on the second term of \eqref{E1}, after changing variables with $W=\ln(T_d)$, $A=\ln(1+\gamma_f)$ and $a=\theta$ the expression of this term can be given as

\begin{equation}\label{RDproof}
\begin{array}{l}
S=\int\limits_{ \ln(T_d)}^\infty  {\Pr \left( {\frac{{{P_d}{h_d}{{\left\| {{X_d}} \right\|}^{ - {\alpha _d}}}}}{{{I_{d,du}} + {I_{f,du}}}} > {e^{{\theta _d}}} - 1} \right)} {\rm{d}}{\theta _d} \\
\approx \int\limits_{ \ln(T_d)}^\infty  {{L_{{I_{d,du}}}}\left( {\frac{{{e^{{\theta _d}}}{{\left\| {{X_d}} \right\|}^{{\alpha _d}}}}}{{{P_d}}}} \right){L_{{I_{f,du}}}}\left( {\frac{{{e^{{\theta _d}}}{{\left\| {{X_d}} \right\|}^{{\alpha _d}}}}}{{{P_d}}}} \right)} {\rm{d}}{\theta _d}\\
= \int\limits_{ \ln(T_d)}^\infty  {\exp \left( { - \pi {{\left\| {{X_d}} \right\|}^{\frac{{2{\alpha _d}}}{{{\alpha _f}}}}}\beta C\left( {{\alpha _f}} \right)e^  {\frac{2{\theta _d}}{{{\alpha _f}}}}} \right)} {\rm{d}}{\theta _d} \\
= -\frac{{{\alpha _f}}}{2}{\rm{Ei}}\left( { -T_d^{\frac{2}{\alpha_f}} \pi {{\left\| {{X_d}} \right\|}^{\frac{{2{\alpha _d}}}{{{\alpha _f}}}}}\beta C\left( {{\alpha _f}} \right)} \right), \\
\end{array}
 \end{equation}
where $\beta=\left( {{\lambda _{du}} +   {{\left( {{P_f}}/{{P_d}} \right)}^{{2}/{{{\alpha _f}}}}}{\lambda _f}} \right) $, and the proof is finished.

\subsection{Proof of lemma 1}
we first drive the coverage probability of nearest F-AP mode with $\alpha_f=4$ and $T_f > 1$, which can be denoted as
\begin{equation}\label{RFproof}
\begin{array}{l}
 P_F^{\alpha_f  = 4}\left( {{T_f},p_c^F} \right) = \frac{1}{{1 + \rho \left( {{T_f},4} \right) + \frac{{{\lambda _{du}}}}{{{p_c^F\lambda _f}}}C\left( 4 \right)\sqrt {\frac{{{P_d}{T_f}}}{{{P_f}}}} }} \\
  = \frac{1}{{1 + \sqrt {{T_f}} \int_{1/\sqrt {{T_f}} }^\infty  {\frac{1}{{1 + {v^2}}}{\rm{d}}v}  + \frac{{\pi {\lambda _{du}}}}{{2{p_c^F\lambda _f}}}\sqrt {\frac{{{P_d}{T_f}}}{{{P_f}}}} }} \\
  = \frac{1}{{1 + \sqrt {{T_f}} \left[ {\frac{\pi }{2} - \arctan \left( {1/\sqrt {{T_f}} } \right)} \right] + \frac{{\pi {\lambda _{du}}}}{{2{p_c^F\lambda _f}}}\sqrt {\frac{{{P_d}{T_f}}}{{{P_f}}}} }} \\
  \mathop  \approx \limits^{\left( a \right)} \frac{1}{{1 + \sqrt {{T_f}} \left[ {\frac{\pi }{2} - \left( {1/\sqrt {{T_f}} } \right)} \right] + \frac{{\pi {\lambda _{du}}}}{{2{p_c^F\lambda _f}}}\sqrt {\frac{{{P_d}{T_f}}}{{{P_f}}}} }} \\
  = \frac{2}{{{{\pi \sqrt {{T_f}} }}\left( {1 + \frac{{{\lambda _{du}}}}{{{p_c^F\lambda _f}}}\sqrt {\frac{{{P_d}}}{{{P_f}}}} } \right)}}, \\
 \end{array}
 \end{equation}
where (a) follows the property of the inverse trigonometric functions that $\arctan(A) \approx A$ if A is smaller than 1, i.e., $T_f \ge 1$.

Then, substituting \eqref{RFproof} into \eqref{E1} with $W=\ln(T_f)$ and we obtain the result.

\end{document}